%% file: main.tex
\documentclass[sigconf,nonacm]{acmart}

\usepackage{xcolor}
\usepackage[ruled,vlined]{algorithm2e}
\usepackage{multirow}
\usepackage{subcaption}
\usepackage[printonlyused]{acronym}
\usepackage{amsmath,amsfonts}
\usepackage{algorithmic}
\usepackage{graphicx}
\usepackage{balance}
\usepackage{mathtools}
\usepackage{booktabs}
\usepackage{textcomp}
\usepackage{listings}
\usepackage{enumitem}
\usepackage{soul}
\usepackage{pifont}
\newcommand{\cmark}{\ding{52}}
\newcommand{\xmark}{\ding{56}}
\usepackage{verbatim}
\usepackage{float}

\AtBeginDocument{%
  \providecommand\BibTeX{{%
    \normalfont B\kern-0.5em{\scshape i\kern-0.25em b}\kern-0.8em\TeX}}}

\lstdefinelanguage{json}{
    basicstyle=\footnotesize\ttfamily,
    numbers=left,
    numberstyle=\scriptsize,
    stepnumber=1,
    numbersep=8pt,
    showstringspaces=false,
    breaklines=true,
    frame=lines,
    backgroundcolor=\color{gray!10},
    literate=
     *{0}{{{\color{orange}0}}}{1}
      {1}{{{\color{orange}1}}}{1}
      {2}{{{\color{orange}2}}}{1}
      {3}{{{\color{orange}3}}}{1}
      {4}{{{\color{orange}4}}}{1}
      {5}{{{\color{orange}5}}}{1}
      {6}{{{\color{orange}6}}}{1}
      {7}{{{\color{orange}7}}}{1}
      {8}{{{\color{orange}8}}}{1}
      {9}{{{\color{orange}9}}}{1}
      {:}{{{\color{purple}:}}}{1}
      {,}{{{\color{purple},}}}{1}
      {\{}{{{\color{blue}\{}}}{1}
      {\}}{{{\color{blue}\}}}}{1}
      {[}{{{\color{blue}[}}}{1}
      {]}{{{\color{blue}]}}}{1},
    breakatwhitespace=false,
    breakautoindent=true,
    postbreak=\mbox{\textcolor{red}{$\hookrightarrow$}\space},
    escapeinside={(*@}{@*)}  % for escaping to LaTeX in listing
}

\begin{document}
%-------------------------------------------------------------------------------

\title{How to Train your Antivirus: RL-based Hardening through the Problem Space}

\author{Ilias Tsingenopoulos}
\authornote{Equal contribution}
\affiliation{%
  \institution{DistriNet, KU Leuven}
  \country{Belgium}
}
\author{Jacopo Cortellazzi}
\authornotemark[1]
\affiliation{%
  \institution{King's College London\\University College London}
  \country{United Kingdom}
}
\author{Branislav Bošanský}
\affiliation{%
  \institution{Gen Digital}
  \country{Czech Republic}
}
\author{Simone Aonzo}
\affiliation{%
  \institution{Eurecom}
  \country{France}
}
\author{Davy Preuveneers}
\affiliation{%
  \institution{DistriNet, KU Leuven}
  \country{Belgium}
}
\author{Wouter Joosen}
\affiliation{%
  \institution{DistriNet, KU Leuven}
  \country{Belgium}
}
\author{Fabio Pierazzi}
\affiliation{%
  \institution{King's College London}
  \country{United Kingdom}
}
\author{Lorenzo Cavallaro}
\affiliation{%
  \institution{University College London}
  \country{United Kingdom}
}

\renewcommand{\shortauthors}{Tsingenopoulos, et al.}

\begin{abstract}
ML-based malware detection on dynamic analysis reports is vulnerable to both evasion and spurious correlations. In this work, we investigate a specific ML architecture employed in the pipeline of a widely-known commercial antivirus, with the goal to harden it against adversarial malware.
Adversarial training, the most reliable defensive technique that can confer empirical robustness, is not applicable out of the box in this domain, for the principal reason that gradient-based perturbations rarely map back to feasible problem-space programs. 
We introduce a novel Reinforcement Learning approach for constructing adversarial examples, a constituent part of adversarially training a model against evasion. Our approach comes with multiple advantages. It performs modifications that are feasible in the problem-space, and only those; thus it circumvents the inverse mapping problem. It also makes it possible to provide theoretical guarantees on the robustness of the model against a well-defined set of adversarial capabilities. Our empirical exploration validates our theoretical insights, where we can consistently reach 0\% Attack Success Rate after a few adversarial retraining iterations.
\end{abstract}

\maketitle

\input{acronyms.tex}

\input{intro.tex}
\input{background.tex}

\input{threat}
\input{approach.tex}
\input{dataset.tex}
\input{binaries.tex}
\input{evaluation.tex}
\input{discussion.tex}

\begin{acks}
% Acknowledgements
This research is partially funded and supported by: the Research Fund KU Leuven;   the Cybersecurity Research Program Flanders; the EU funded project KINAITICS (Grant Agreement Number 101070176); a Google ASPIRE research award; EPSRC Grant EP/X015971/1.
\end{acks}

\balance
\bibliographystyle{ACM-Reference-Format}
\bibliography{references}
\clearpage
\input{appendix}

\end{document}

%% file: acronyms.tex
\acrodef{ML}[ML]{machine learning}
\acrodef{DL}[DL]{deep learning}
\acrodef{AI}[AI]{artificial intelligence}
\acrodef{RL}[RL]{reinforcement learning}
\acrodef{AML}[AML]{adversarial machine learning}
\acrodef{MDP}[MDP]{Markov Decision Process}
\acrodef{DNN}[DNN]{deep neural network}
\acrodef{RNN}[RNN]{recurrent neural network}
\acrodef{SVM}[SVM]{Support Vector Machine}
\acrodef{RL}[RL]{reinforcement learning}
\acrodef{NLP}[NLP]{natural language processing}
\acrodef{LLM}[LLM]{large language models}
\acrodef{GAN}[GAN]{generative adversarial network}
\acrodef{IID}[IID]{independent and identically distributed}
\acrodef{OOD}[OOD]{out-of-distribution}
\acrodef{HW}[HW]{Hardwaere}
\acrodef{SW}[SW]{Software}
\acrodef{AE}[AE]{Adversarial Example}
\acrodef{FPS}[FPS]{frames per second}
\acrodef{PPO}[PPO]{Proximal Policy Optimization}
\acrodef{HMIL}[HMIL]{Hierarchical Multiple Instance Learning}
% abbreviations
\acrodef{wrt}[\emph{w.r.t.}]{with respect to}
\acrodef{st}[\emph{s.t.}]{such that}

%% file: intro.tex
\section{Introduction}
Adversarial examples have been extensively explored in numerous domains. While the fundamental challenge they pose to learning is fascinating on its own scientific accord, they also represent a significant threat to the widespread adoption of machine learning solutions in security-critical scenarios.
At the same time, there is an eagerness by many cybersecurity vendors to adopt deep-learning approaches for malware detection and classification, for two main reasons.
First, adversarial examples have existed in security domains long before they emerged in \ac{AI}, as straightforward ways for adversaries to adapt to signature- and rule-based decision making---for instance variants from the same malware family.
Secondly, while rule-based decisions can claim near zero false positives, they perform poorly in distributional shifts~\cite{upadhyay2021towards}, which are further intensified under adversarial agency.
Hence, the promise that \ac{AI} brings is to learn general and ideally robust rules of inference in the presence of distributional shifts, both naturally and adversarially induced.

Adversarial malware examples have been predominantly explored in the feature space and on static analysis features.
With the prevalence of obfuscation, polymorphic strains, and packing, static analysis is becoming less relevant for malware detection~\cite{aghakhani2020malware}.
Behavioral analysis can be a very effective tool as it circumvents the challenges of obfuscated and packed binaries by exposing their runtime behavior, considered a reliable fingerprint of malicious intent.
The representation of this behavior can range from API call sequences~\cite{tian2010differentiating, rosenberg2018generic}, to fine-grained modeling of process behavior~\cite{continella2016shieldfs}, and fully textual reports~\cite{mandlik2022jsongrinder}.
As dynamic analysis tools gain traction and widespread adoption, novel vulnerabilities are introduced: adversarial malware behavior that evades detection while preserving the intended functionality.
Notably, this behavior can be learned as a direct consequence of dynamic analysis itself, the very process that is supposed to thwart it.

In other domains, for example image classification, the state-of-the-art defense against evasion is adversarial training \cite{madry2017towards, tramer2020adaptive}.
We can think of adversarial training as a process of generating counterfactuals to the data the model is training on: if these features were different under certain constraints, would that still be the same example?
While the most effective approach for model hardening, adversarial training remains largely empirical: it is still intractable to provide guarantees on the operational robustness for any but the most trivial models.
Furthermore, it is not directly applicable in domains that exhibit a large problem- to feature-space gap~\cite{pierazzi2020intriguing}.
In dynamic analysis-based malware detection, there are several degrees of separation from the original program to the decision of the model.
Not only it is challenging to map any gradient-based perturbations computed on the feature representation back to \textit{viable} programs, but as we show in this work it is also irrelevant when defending against real-world attackers.

The imperceptible perturbation model of adversarial examples bears little resemblance to adversarial malware, which always manifest in the problem space and under a very different set of constraints.
As adversarial examples are essentially out-of-distribution, it is impractical to proactively harden a model against \textit{all} possible variations---especially in domains where they are not constrained by similarity to the original ones.
It is well established that adversarial training can reduce the performance on clean samples; a lesser known property is that it can also reduce the robustness to attacks \textit{other} than the one it was performed with~\cite{gilmer2019a}.
As hardened models are not universally robust but instead biased towards specific kinds of distributional shifts, adversarial training has the potential to harm robustness in real-world settings.
It then follows that in the domain of malware detection, adversarial training is important to be performed with the \textit{full range} of attacks that are \textit{possible} through the problem space---and that the threat analysis has rendered as such---and possibly \textit{only} against those.

In this paper, we make the following contributions:
\begin{itemize}
    \item We first obtain a wide range of Windows binaries spanning from 2012 to 2020, from various malware families to verified goodware; then we execute them in CAPEv2 to generate dynamic analysis reports. To foster research on dynamic analysis, we release our dataset in the following URL: https://paperswithcode.com/dataset/autorobust
    \item We justify theoretically how, under certain conditions, adversarial training through the problem-space is advantageous to gradient-based approaches like PGD~\cite{madry2017towards}, and define the concise set of transformations that are available for behavioral representations.
    \item We introduce AutoRobust, a novel RL-driven methodology for performing adversarial training and observe that, within several attack iterations followed by retraining, the model consistently moves from entirely vulnerable ($\sim$100\% ASR) to fully robust ($\sim$0\% ASR). Notably, AutoRobust does not assume anything about the underlying model or the adversarial capabilities apart from them being functionality-preserving. Thus, it is applicable to \textit{any} problem-space transformations that related works investigate.
    \item The key takeaways of our work are the ability of our approach to bypass an ML component of a commercial antivirus, as well as the importance of incorporating adversarial training through the problem-space. To enable reproducibility and follow-up work, we open-source our codebase: https://github.com/s2labres/AutoRobust
\end{itemize}

%% file: background.tex
\section{Background \& Related Work}
\label{sec:background}
In this section we discuss prior work on the domain and identify open problems and research gaps, setting in this way the ground for our approach.

\subsection{Malware detection}
Prior research has employed a wide range of classifiers to identify malicious software, in diverse contexts~\cite{abusitta2021malware}. 
On Windows, the context we scope to, related work has operated on representations ranging from the whole binary \cite{raff2018malware}, to control flow graphs \cite{alasmary2020soteria}, and static analysis \cite{christodorescu2003static}.
Static analysis however has clear limitations in malware detection and often dynamic analysis is employed as a complementary approach.
Recently, it has been demonstrated that ML models built solely on static analysis are not able to distinguish between benign and malicious samples that are packed \cite{aghakhani2020malware}.
This work reviews a wide range of ML models built on static analysis and concludes that they detect packing rather than malice.
Given encryption, polymorphic strains, and packing, the efficacy of static analysis or signature-based detection deteriorates.

To capture the actual behavior of a program, researchers and analysts have relied on dynamic analysis \cite{cornelissen2009systematic}.
Today there are multiple vendors that provide sandboxes for dynamic analysis; as enterprise or cloud solutions but also open source like CAPEv2 \cite{CAPEv2,CUCKOO,DRAKVUF}, used both in academia and in industry.
Sandboxes typically generate a structured textual description of the full program behaviour, from files accessed and API sequences~\cite{rosenberg2018generic,rosenberg2021sequence,pennington2014glove} to any malicious or evasive techniques used~\cite{trinius2009malware,galloro2022systematical,maffia2021longitudinal}.
Malware detection has historically relied heavily on signatures and indicators of compromise, with the evident shortcoming in the inability to generalize to new strains or zero-day malware.
To thwart such threats, analysts require models that can extrapolate from the observed behavior of existing malware to new samples.

However, going through dynamic analysis reports is a highly involving and time-consuming process, so given the immense amount of data that can be generated there is a great incentive to automate this process.
Learning directly from raw data inputs by obviating the need for feature engineering has been very effective in numerous applications; manual feature engineering is not only labor-intensive, but also prone to introducing human bias.
By processing the raw input this occurs, ML models have the capacity to construct important features for inference and provide further insights to experts, beyond any a priori notions on what constitutes malware behavior.
To this end, recent work has introduced an approach that fully automates this process, named \ac{HMIL}~\cite{mandlik2022jsongrinder,pevny2017using}.
Its relevance and significance lie in the fact that it can be directly applied to the results of dynamic analysis, as it renders structured text reports of arbitrary size end-to-end differentiable.
This makes it possible to be integrated with any ML pipeline, by enabling learning on JSON data in their raw form that can span hundreds of thousands of tokens, a challenge even for the largest of LLMs.

\subsection{Problem-Space Attacks}
While adversarial attacks have been investigated in a multitude of domains~\cite{biggio2018wild}, on malware they were initially explored in the feature-space~\cite{grosse2017adversarial, hu2017generating, stokes2018attack}.
However there is a growing body of works that performs attacks on the full processing pipeline, even up to perturbing the original binary \cite{anderson2018learning, pierazzi2020intriguing, labaca2021realizable, tsingenopoulos2022adaptive,  galloro2022systematical, park2020survey}, using approaches like process chopping \cite{ispoglou2016malwash}, return-oriented programming \cite{ntantogian2019transforming}, bytecode modifications \cite{burr2021improving}, and opcode n-grams \cite{li2020adversarial}.
Furthermore, adversarial malware can be generated on other models and still transfer to (that is successfully evade) a malware classifier under consideration~\cite{demontis2019adversarial}.
These works illustrate how a brittle feature representation can render a model vulnerable to evasion, especially when the model does not process the actual program behavior.
Performing attacks through the problem space and the ensuing constraints that this imposes is an accurate representation of the actual level of threat; an important step as only a realistic assessment of that threat can build effective defenses.

In domains like malware, adversarial attacks are much more difficult to carry out, precisely because of the problem-feature space gap.
In real-world settings, problem and feature spaces are well apart: any approach that computes perturbations in the feature space is always confounded by the task of mapping them back to the problem space in a feasible manner~\cite{sheatsley2020adversarial}.
In the malware domain specifically, the non-invertibility of the mapping from features to programs, the preservation of functionality, and side-effects on the feature representation when valid perturbations are made in the problem space, make this task non-trivial.
To that end, several previous works focus on the problem-space generation and verification of functionality-preserving adversarial malware~\cite{demetrio2021adversarial, demetrio2021functionality, labaca2021aimed}.
In constructing adversarial malware it would therefore be advantageous if the agency that generated them was situated \textit{directly} on the problem space. 

\subsection{Defenses \& Mitigations}
While ML-based malware detection is typically one of the various modules for detection a commercial antivirus solution employs~\cite{botacin2022antiviruses}, the defensive aspect has received relatively less attention compared to the offensive one.
Approaches range from feature-space hardening \cite{galovic2021improving}, to combining distillation with adversarial training~\cite{rathore2021robust}, to defenses tailored to API call-based models~\cite{rosenberg2021sequence}.
In the defensive literature several important principles have been investigated, from using model ensembles to the necessity of adversarial training~\cite{li2021framework}, however such defenses are predominantly performed on the feature representations that become inputs for the ML models, or without consideration to problem space constraints.

In the long history of proposed and eventually broken defenses against adversarial examples~\cite{tramer2020adaptive}, mainly one has stood the test of time: adversarial training~\cite{madry2017towards}.
The objective of adversarial training is to generate adversarial examples during training: given dataset $D = {(x_i, y_i)}^{n}_{i=1}$ with classes $C$ where $x_i \in \mathbb{R}^d$ is a clean example and $y_i \in {1,..., C}$ is the associated label, the goal is to solve the following \emph{min-max} optimization problem:

\begin{equation}
    \underset{\phi}{\operatorname{min}} \mathbb{E}_{x_i, y_i\sim D} \underset{\Vert \delta_i \Vert_{L_p} \leq \epsilon}{\operatorname{max}} \; \mathcal{L}(h_{\phi}(x_i + \delta_i), y_i)
\label{eqn:adv_train}
\end{equation}

\noindent where $x_i + \delta_i$ is an adversarial example of $x_i$, $h_\phi : \mathbb{R^d} \rightarrow \mathbb{R^C}$ is a hypothesis function and $\mathcal{L}(h_\phi(x_i + \delta_i), y_i)$ is the loss function for the adversarial example $x_i + \delta_i$.
The inner maximization loop finds an adversarial example of $x_i$ with label $y_i$ for a given $L_p$-norm (with $L_p \in \{0,1,2,\inf\}$), such that $\Vert \delta_i\Vert_{l} \leq \epsilon$ and $h_\phi(x_i + \delta_i) \neq y_i$.
The outer loop is the standard minimization task typically solved with stochastic gradient descent.

A widely accepted principle in robustness to distributional shift literature~\cite{lanckriet2002robust, geirhos2018generalisation} is that the lack of robustness can be largely attributed to spurious correlations, that is the model latching onto superficial information and artifacts manifesting as causal links between the dependent and independent variables.
Analogously to computer vision, in the malware domain this information can appear unintuitive or downright irrelevant to us, for instance the specific filenames used.
Yet such artifacts can still prove useful for generalization \textit{within} the \ac{IID} settings, as these are defined by the currently available data.

Adversarial training is, in essence, an exploration on \ac{OOD} variations of these data, under a set of constraints.
As such, if this exploration is performed without regard to what variations are possible and which constraints are there, it is difficult to ascertain what the effects on the clean and robust accuracy will be.
Previous works have shown that undisciplined adversarial training can provide limited robustness or even make the model more vulnerable to other transformations or constraints~\cite{hendrycks2019benchmarking, dyrmishi2023empirical}.
If we look beyond $\ell_p$ norm perturbations and take a broader view on robustness, it becomes clear that one cannot claim that adversarially trained models are robust \emph{in general}, but rather biased towards \textit{specific} distortions as these are generated by specific adversarial capabilities. 
It is entirely conceivable then that adversarial training can harm not only performance but robustness too in real-world settings.

If adversarial training has lackluster performance or even the potential to diminish robustness when performed irrespective of transformations, then it should be performed under those and \textit{only} those that one wants to defend against.
In the context of dynamic analysis, this would elicit an exhaustive search to identify the concise set of transformations an adversary can make in the problem space.
As dynamic analysis documents program behavior as it is expressed over time, this set of transformations can result to a very diverse set of possible adversarial examples, especially if one does not constrain the perturbation amount.

\subsection{Research Gap}
While adversarial training remains the most promising path to explore in defenses, the meaningfulness of $L_p$-norms and $\delta$ perturbations as constraints applies mostly in the computer vision domain, where visual similarity is important and where the majority of adversarial attacks and defenses were initially researched and performed.
Every other domain, including computer vision but in real-world settings~\cite{sharif2019general}, will have slightly different to completely unrelated set of constraints.
Generating adversarial malware in particular, demonstrates two key idiosyncrasies compared to computer vision.
Modifications are more difficult to perform, as there are hard constraints stemming from the preservation of the semantics and the original program functionality.
However, only soft constraints apply on the total amount of modifications, as there is no strict requirement of resemblance to the original binary.
Previous work has already explored adversarial training on malware classifiers and under such constraints~\cite{lucas2023adversarial}.
They however focus on hardening models that operate on the raw-byte level of a binary, specifically the MalConv and AvastNet architectures.

In this work we focus on hardening a model that operates on reports generated by executing the malware in a sandbox and generating a dynamic analysis report.
In this context then, the natural form that data occur in is textual; for this representation adversarial training, as is conventionally performed, can not be applied out of the box due to the following reasons:
\begin{itemize}
    \item White box techniques, like FGSM and PGD, maximize the inner loop of the loss function \eqref{eqn:adv_train} during training: this often results in feature-space perturbations that do not map back to valid program behavior.
    \item Dynamic analysis reports exhibit a huge variance in size (from ten tokens to hundreds of thousands), and the possible modifications on them are in theory unbounded. As not all possible adversarial examples are interesting or useful, it is better to reground the problem and make the problem tractable by focusing on those adversarial examples that are \textit{feasible} through the problem space.
    \item Adversarial training remains an empirical hardening technique. It would be preferable if through a rigorous threat analysis we could show that against a specific and representative set of problem-space modifications, we could achieve a guaranteed level of robustness.
\end{itemize}

%% file: threat.tex
\section{Threat Model}
\label{sec:threat}

The methodology we conceive and develop in this work is a general framework for hardening ML-based classifiers $\mathcal{M}$ against adversaries that employ well-defined problem-space capabilities $\mathcal{C}$ for evasion.
Under this formulation, it is independent of the underlying model $\mathcal{M}$ and capabilities $\mathcal{C}$, and can thus accommodate a wide variety of them.
In this section, we explicit our assumptions before we proceed to the concrete details of the design and our empirical evaluation.
The threat model we consider focuses on the dynamic analysis report representation and is delineated as follows:

\begin{itemize}
    \item[--] \textbf{Assets:} Trained and deployed model $\mathcal{M}$ with corresponding weights ${w}$.
    \item[--] \textbf{Adversaries:} Malware coders or malware-as-a-service (MaaS) actors.
    \item[--] \textbf{Knowledge:} Given training data $\mathcal{D}$, feature representation $\mathcal{F}$, learning algorithm $g$, model $\mathcal{M}$, and explainer $\mathcal{E}$, we assume a Perfect Knowledge attacker that knows all $\theta_{PK} = (\mathcal{D},\mathcal{F},g,\mathcal{M}, \mathcal{E})$ \cite{biggio2018wild}.
    \item[--] \textbf{Goal:} Generate adversarial malware, that is render the malware binary evasive with respect to the model $\mathcal{M}$ while preserving its original functionality.
    \item[--] \textbf{Capabilities}: Make transformations on the binary that are reflected in the dynamic analysis report, as these are described in \autoref{sec:trans}. No dataset poisoning capabilities.
    \item[--] \textbf{Defenses:} Use our framework to generate report variants for adversarially training the model $\mathcal{M}$.
\end{itemize}

In our threat model, assuming the perspective of the adversary, the objective is to produce an execution trace that would create a dynamic analysis report that is not detected as malicious by the target model $\mathcal{M}$.
While this can be achieved by modifying the source code of the program in ways that preserve the original functionality~\cite{ming2017impeding}, we consider this capability out of scope.
A more representative and realistic threat is the attacker modifying the compiled binaries directly; we elaborate on how this can be achieved in Section \ref{sebsec:binaries}.
On the other hand, the defender's objective is to increase the robustness of model $\mathcal{M}$ to evasion, but similarly they have no access to the source code so all transformations are assumed to be performed on the compiled binary.
This implies that all possible transformations $\Lambda$ on the original binary will map to a subset of all the possible transformations in the feature representation $\Phi$ (see Fig. \ref{fig:grad}); the adversary cannot modify the representation in the input space of the model $\mathcal{M}$ at will.

%% file: approach.tex
\section{AutoRobust}
In this section we present AutoRobust, our general framework and methodology to harden ML-based models against problem-space attacks.
Given known adversarial capabilities, we theoretically justify how our approach is preferable to gradient-based white-box hardening, and also demonstrate it empirically.

\begin{figure*}[ht]
\centering
\includegraphics[width=1.6\columnwidth]{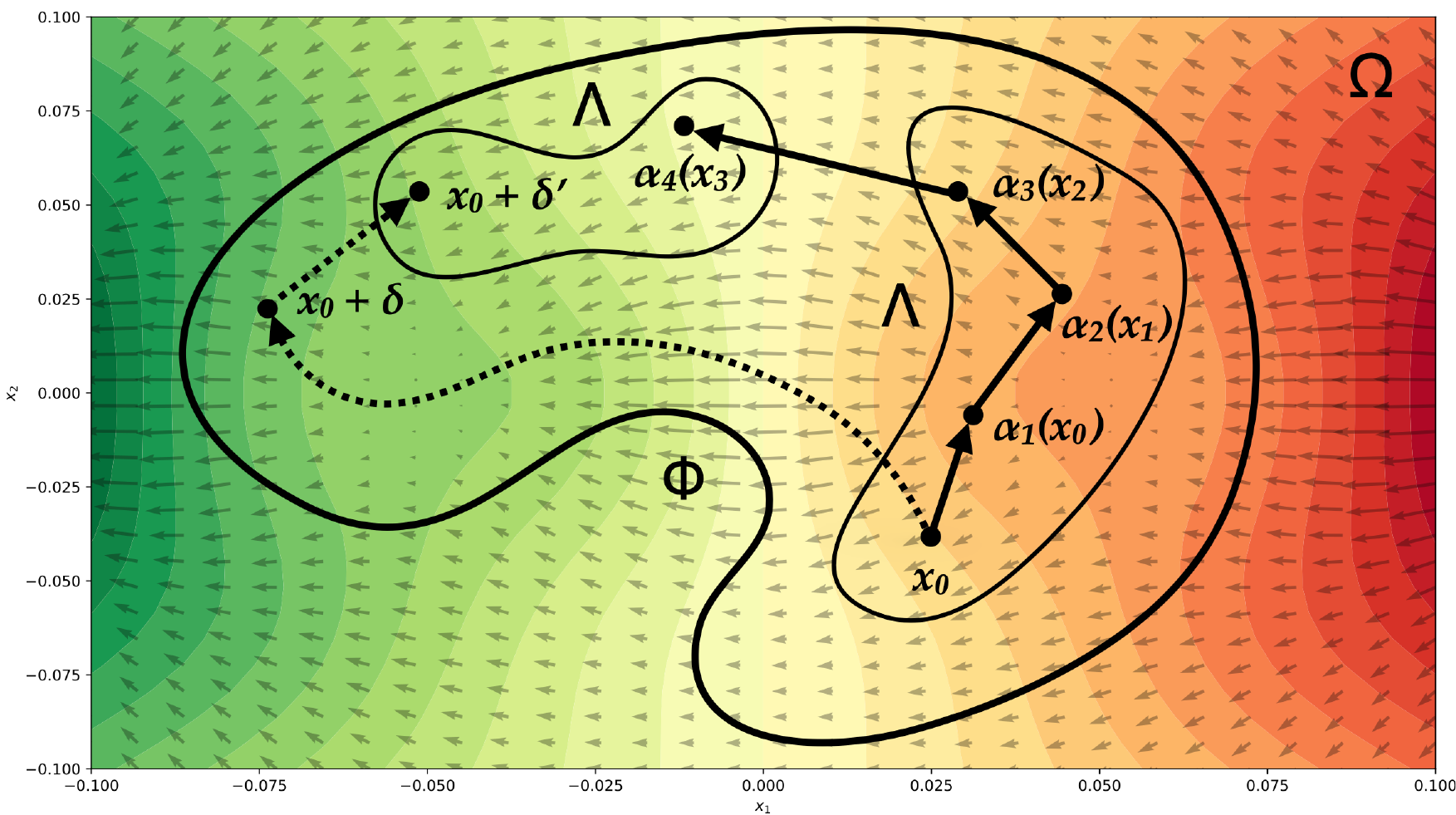}
\caption{\textbf{Comparison between traditional gradient-based attacks and AutoRobust.} The dotted path shows a typical gradient-based attack: first perturbing to $\mathbf{x} + \boldsymbol{\delta}$ then projecting to $\mathbf{x} + \boldsymbol{\delta}^\prime$ in the \emph{feasible} problem-space $\Lambda$. Our approach (dense path) that employs transformations $\alpha_t$ in succession, moves by definition \textit{only} within the feasible problem space $\Lambda$. The background displays a gradient field over the value of the discriminant function $h(\mathbf{x})$, with negative values (green) for the target class. The thick solid area $\Phi$ represents the feasible feature-space, while the areas denoted by $\Lambda$ represent the feasible problem-space mapped to $\Phi$.}
\label{fig:grad}
\end{figure*}

\subsection{Approach}
Let $\Omega$ be the latent space of a model $\mathcal{M}$ where inputs are mapped, for example the last layer before output.
During training, the model learns such a latent space implicitly, with benign and malicious areas being well separated.
Gradient-based attacks compute perturbations by defining an adversarial loss on the model and then backpropagating through it.
With  $\Phi \subset \Omega$ we represent all the possible representations in the \textit{feature-space}.
While every input sample $x$ maps to a point in $\Phi$, not every point in $\Phi$ maps to a valid input sample: the mapping that model $\mathcal{M}$ performs is injective but not surjective.
We denote the space of valid input samples as $\Lambda \subset \Phi$.
Consequently, following gradient-based perturbations will inadvertently result to points that are not in $\Lambda$ and thus invalid, something that will require more computation to correct by mapping them back to $\Lambda$, and often deviating from the gradient direction itself.

Consider now an alternative to gradient-based perturbations: problem-space modifications from a capability set $\mathcal{C}$ that by definition produce only valid programs.
Then the resulting reports are \textit{guaranteed} to fall within $\Lambda$.
If we denote as $S$ the embedding in $\Lambda$ of an original malware sample, then making a transformation $T \in \mathcal{C}$ will produce a valid sample $S'$ that is always in $\Lambda$.
Then we can construe the adversarial task as a \ac{MDP} to be solved: if we take the latent subspace $\Lambda$ as the state space, and the transformations $\alpha_t$ as the action space, the goal is to sequentially apply transformations until the model $\mathcal{M}$ is evaded.
Notably, this process can be optimised in gradient-based manner too \cite{sutton1999policy}, while always staying on the valid program space.
An intuitive illustration of our approach is included in Figure \ref{fig:grad}.

In the domain of adversarial malware, the transition from gradient-based methodologies like PGD~\cite{madry2017towards} or C\&W~\cite{carlini2017towards} to problem-space attacks like GAMMA~\cite{demetrio2021functionality} and AIMED-RL~\cite{labaca2021aimed} guarantees the preservation of semantics and functionality.
AutoRobust, which as a framework is independent of the underlying modifications performed, comes with an additional benefit: if an agent with capability set $\mathcal{C}$ cannot find a feasible path to the adversarial side of the decision boundary of $\mathcal{M}$, then this state $S^*$ is either unreachable under $\mathcal{C}$, or does not exist, and thus the agent is unable to evade the model.
In this manner we can show that for a number of samples with starting states $s_0$, the model $\mathcal{M}$ is robust to the specific $\mathcal{C}$ with probability $p$, denoted as p-Robust.

\begin{theorem}[Problem-Space p-Robustness]
Given model $\mathcal{M}$ and adversary $\mathcal{V}$ with problem-space capabilities $\mathcal{C}$, $\mathcal{M}$ is robust to problem-space evasion with probability $p$ if and only if the expected reward of the optimal policy $\pi^*(a|s)$ in the corresponding MDP is $1-p$.
\label{th:pspr}
\end{theorem}

\begin{proof}
Let us assume a scalar reward for the adversary: 1 if they evade at time step $t > 0$ with the episode terminating afterwards, or 0 if they cannot evade up to a maximum time horizon $\mathcal{T}$.
The objective of the adversary is to maximize the reward of their policy $\pi_\theta^\mathcal{V}(a|s)$, parameterized by $\theta$, over a number of episodes $N$.
Then the expected value of the reward $\mathcal{R}$ is:

\begin{equation}
    \mathbb{E}_{\pi_\theta^\mathcal{V}}[\mathcal{R}] = \frac{1}{N}\sum_{i=1}^{N}\mathcal{R}
\label{eqn:expect}
\end{equation}

The optimal policy $\pi^*(a|s)$ is by definition the one that maximizes this expectation, and from the Policy Gradient Theorem \cite{sutton1999policy} we can compute it in a gradient-based manner.
If we bijectively map actions $a \in \mathcal{A}$ to transformations $T$ in the capability set $\mathcal{C}$, and states $s \in \mathcal{S}$ to the latent space $\Lambda$ of $\mathcal{M}$, then it is sound to construe this adversarial task as the corresponding MDP: transformations on current state $s$, deterministically lead to only one successor $s'$.
Additionally, the Markov Property is preserved as this transition depends solely on $s$ and transformation $\alpha_t \in T$. 
Now if we assume the value of the expectation is $1-p$, and given that $\mathcal{R} \in \{0,1\}$, from \eqref{eqn:expect} we can show that:

\begin{equation}
\begin{aligned}
    \frac{1}{N}\sum_{i=1}^{N}\mathcal{R} = \frac{1}{N}\sum_{i=1}^{N}(0 \cdot P(ben) + 1 \cdot P(adv)) = P(adv) = 1 - p
\label{eqn:expect2}
\end{aligned}
\end{equation}

It follows that the probability $P(ben)$ of the model $\mathcal{M}$ being robust to the adversary is $p$.
The converse is also straightforward to show: if we assume that a model is p-Robust, then for $\mathcal{V^*}$ an optimal adversary, $P(ben) = p$, and from \eqref{eqn:expect} the expectation is $1-p$.
\end{proof}

We have shown how one can use RL to find optimal transformation policies in the problem-space.
Afterwards, we proceed with the outer loop of Eq. \eqref{eqn:adv_train} unchanged: the newly discovered adversarial examples together with the original, clean ones are used to retrain the model.
This cycle is then repeated, where it becomes successively more difficult for the RL agent to construct adversarial examples.
While the degree of robustness is estimated on the validation set of the generated adversarial examples, these are generated under the same agent policy which means they will fall under the same distribution; thus they cannot be a representative measurement of the actual level of robustness.
The proper way to assess that is by training the agent anew in a hold-out test set.
Then given sufficient horizon to converge, the resulting optimal policy will indicate the level of p-Robustness, for this specific model $\mathcal{M}$ and this specific capability set $\mathcal{C}$.

\begin{figure*}[!t]
\centering
\includegraphics[width=1.7\columnwidth]{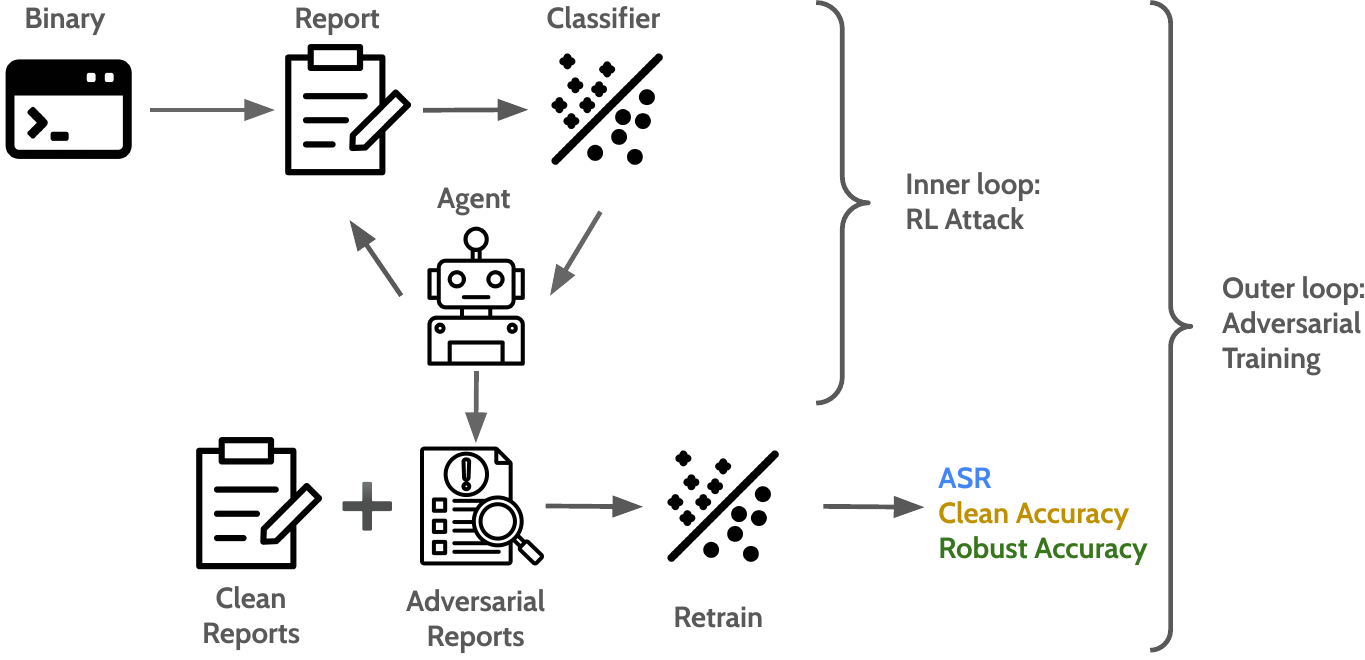}
\caption{Schematic depiction of the AutoRobust pipeline. In the inner loop, the RL agent attacks the model and generates adversarial reports. In the outer loop, the model is retrained with standard minibatch gradient descent on clean \textit{and} adversarial reports.}
\label{fig:autorobust}
\end{figure*}

\textbf{Generality.} It should be noted that Theorem \ref{th:pspr} does not assume anything about the discriminative model $\mathcal{M}$, while for the capability set $\mathcal{C}$ the only assumption we make is that transformations in this set are functionality preserving.
Then their application and composition will result in valid programs whose embeddings that are \textit{always} in $\Lambda$.
This indicates that our methodology is applicable to \textit{any} functionality preserving transformations in the problem-space; for instance to transformations directly on the binary that a growing literature investigates and performs~\cite{demetrio2021adversarial, demetrio2021functionality, labaca2021realizable, labaca2021aimed}.

\textbf{Hardening.} We assert that in order to defend against realistic adversarial attacks, i.e. those that can be carried out from the problem-space, it is insufficient to perform adversarial training in the feature-space and then show a degree of empirical robustness for some arbitrary perturbation budget.
This is because, unlike images, dynamic analysis reports can be of arbitrary size: adversarial training that does not take into account what an adversary \textit{can} do will inadvertently explore inconsequential spaces around the decision boundary that do not affect evasion.
Instead, as we will also demonstrate in our evaluation (Section \ref{sec:eval}), it is preferable to define the specific modifications that are feasible in the problem-space -- a result of thorough threat modeling -- and then construct adversarial examples with them.

Our observation is further corroborated through the results of Dyrmyshi et al.~\cite{dyrmishi2023empirical}, where they discover that in domains where a problem-feature space representation gap exists, realistic attacks can be much more effective in hardening models against evasion.
More concerningly, adversarial training with feature-space, unrealistic attacks has the capacity to actually make a model \textit{more} vulnerable to adversarial attacks, as Hendrycks et al. have demonstrated~\cite{hendrycks2019benchmarking}.
In conjunction, these two properties indicate an important implication for all security-critical domains: adversarial training should be performed with \textit{all} those and \textit{only} those transformations that are feasible, threatening, and can \textit{actually} be carried out through the problem-space.

\subsection{Explanation-guided Hardening}
\label{sec:explainer}
Conceptually, our approach can be viewed as a practical and automated method for identifying useful but non-robust features \cite{ilyas2019adversarial}.
Until now this was not straightforward to perform in domains where a distinct gap exists between problem and feature-space, as this identification requires an expert in-the-loop.
Such useful but brittle features are commonly described as shortcuts or spurious correlations~\cite{geirhos2020shortcut}.
By passing to the RL agent the feasible modifications and \textbf{only} those, the agent will optimize towards changing the decision of the model, in essence discovering counterfactuals to what the model has learned: from low order features, e.g. specific filenames, to higher order concepts, e.g. relative size between features~\cite{kim2018interpretability}.

This exploration is intractable to be exhaustively performed in continuous or unconstrained in size domains; it can however become more efficient when some bias is introduced, for instance towards the most important features for classification~\cite{casper2022robust}.
Our task in this work is to harden against evasion an end-to-end differentiable ML model, known as \ac{HMIL} \cite{pevny2017using}, which ingests and classifies raw dynamic analysis reports.
This model is accompanied by an explainability tool $\mathcal{E}$, specifically built for classifiers trained on raw hierarchical structured data~\cite{pevny2022explaining}.
The explanations returned by this tool consist of the minimal subset of input, in our case a sub-tree of the report, that receives the same classification as the complete input sample.
The process of acquiring an explanation, given a sample to classify, consists of two phases: sub-tree ranking and sub-tree searching.
First, all sub-trees are heuristically ranked based on their importance for the final classification.
Then, a minimal sub-tree of the sample is searched such that its evaluation by the model remains within a threshold $\tau$ of the original.
The end goal is to identify and distill the essential parts of the sample that influence the classification decision.

After this process concludes, a minimal set of entries from the dynamic analysis report is returned, something that is instrumental in our counterfactual investigation.
As the number of entries in these reports can be arbitrarily large (in the hundreds of thousands), we are interested in identifying the most relevant and influential ones.
Then a policy for modifying these reports can be enhanced through these explanations by precisely pointing towards what, according to the model under attack, is important for its decision.
Note however that an entry pointed to by the explainer might not be modifiable in a way that it preserves the original functionality of the binary; AutoRobust performs only those modifications that preserve the functionality.
After a sample has evaded detection, or a modification limit has been reached, it is added in the pool of adversarial examples.
From there, to make the model robust to evasion and spurious correlations, we introduce these counterfactual examples to the training set -- with their correct, original label -- and together with the clean samples we perform adversarial retraining of the model.

%% file: dataset.tex
\section{Methodology}
\label{dataset}

In this section we describe in detail our methodology and implementation for evaluating AutoRobust.
This involves the dataset and binaries used, the dynamic analysis performed, the representation of that data, as well as the concrete transformations used on these representations.

\subsection{Dataset}
Unlike signature-based or rule-based detection, to build discriminative models for classifying between malicious and benign programs, we require samples from both.
The paradox being that while malware are fairly easy to find, goodware are much more difficult to systematically procure.
We source our goodware as they do in Dambra et al.~\cite{dambra2023decoding}: through the community-maintained packages of Chocolatey~\cite{Chocolately} they create a dataset that spans 2012 to 2020.
Regarding malware, VirusTotal~\cite{VirusTotal} provided us with a dataset of Portable Executable from 2017--2020 that they released for academic purposes.
Since we could not execute malware with network connectivity for safety and ethical concerns, we selected well-known representative families that would exhibit their malicious behavior under dynamic analysis even in the absence of connectivity: virlock, vobfus, and shipup.
We verified that these families are sufficiently different from each other with vhash (resulting in a maximum overlap of 2,5\%), while we preserve their distinct labels in order to evaluate also in a multi-class scenario, which was deemed of high interest by our industrial partner. 
In total, the dataset we managed to assemble contains 26,200 PE samples: 8,600 (33\%) goodware and 17,600 (67\%) malware.

All samples were executed in a series of Windows 7 VMware virtual machines, each with 4GB RAM, 4 vCPU, 128GB HDD, without network connectivity due to policy/ethical constraints.
These were orchestrated through CAPEv2 \cite{CAPEv2}, a maintained open-source successor of Cuckoo \cite{CUCKOO}. 
This framework is widely used for analyzing and detecting potentially malicious binaries by executing them in an isolated environment and observing their behavior without risking the security of the host system.
Every sample was executed for 150 seconds, an empirical lower bound on the time required to gather the full behavior of samples~\cite{kuchler2021does}.

As is considered best practice in the literature~\cite{maffia2021longitudinal}, we tried to mitigate evasive checks by running the VMwareCloak~\cite{VMwareCloak} script to remove well-known artifacts introduced by VMware.
Moreover, we populated the filesystem with documents and other common types of files, to resemble a legitimate desktop workstation that malware may identify as a valuable target~\cite{miramirkhani2017spotless}. 
The dynamic analysis output is a detailed report with information on the syscalls\footnote{Without loss of generality, we use Windows APIs, Windows Native APIs, system calls, and syscalls interchangeably.} invoked by the binary and all the relative flags and arguments, as well as all interactions with the file system, registry, network, and other key elements of the operating system.
Malware analysis and all other experiments ran on an infrastructure with 32 cores Intel(R) Xeon(R) Gold 6140 CPU @ 2.30GHz and 256GB of RAM.

\subsection{Representation}
\label{sec:representation}
While the reports generated by CAPEv2 analysis consist of detailed information on the execution behavior of the binary, they vary considerably in size and may contain redundant or uninformative parts. 
As input space for the HMIL models that we train, we use the summary part of these reports.
This summary is comprised of a concise description of all interactions the binary has with the system at the system call granularity level, and is composed of the following 13 categories:
\begin{itemize}
    \item \textbf{Files / Read Files / Write Files / Delete Files}, containing the sets of all, read, written, and deleted files during execution respectively.
    \item \textbf{Keys / Read Keys / Write Keys / Delete Keys}, containing the sets of all, read, written, and deleted registry keys during execution respectively.
    \item \textbf{Executed Commands}, being the set of commands executed.
    \item \textbf{Resolved APIs}, being the set of external API calls that have been resolved.
    \item \textbf{Mutexes}, being the set of mutex names that have been invoked.
    \item \textbf{Created Services / Started Services}, being the sets of service names that have been created and started respectively.
\end{itemize}

This representation is delivered in JSON format, which includes a distinct key for each category described, with the corresponding values consisting of unique entries.
We should note here that in the summary report of CAPEv2, in each category individual entries are listed only once.

\subsection{Transformations}
\label{sec:trans}
The explicit goal of our methodology is to identify counterfactuals on the input space of the model that indicate spurious correlations and brittle features that can be used for evasion, otherwise known as \textit{shortcuts}~\cite{geirhos2018generalisation}.
For example, our intuition indicates that a specific filename should be an irrelevant feature as it could be replaced by any other.
As we are interested in defending against realistic threats, we need to exhaustively determine all the perturbations on the feature representation $\mathcal{F}$ of a sample that can be performed by the modifications that are permissible in the problem space and on the original program~\cite{pierazzi2020intriguing}.

By permissible here we mean transformations that preserve the original functionality of the program.
Addition is allowed across all categories as it is possible to introduce new entries without affecting the malicious functionality.
Replacement is also feasible for all categories, with the exception of replacing APIs as it represents a potential risk for breaking the functionality.
We consider feature removal in all categories a non-viable transformation; it is theoretically feasible, but not without refactoring or rewriting the source code, something we deem out of our threat model and our scope in this work.
The set of transformations that are feasible through the problem-space is summarized in Table \ref{table:capabilities}.

All report modifications are performed in a consistent manner and properly reflected where necessary.
The structured report that represents the binary behavior contains interdependencies among the features categories that have to be respected when modifications take place; otherwise these modifications would result in a dynamic analysis report that is unrepresentative of an actual executing binary.
For instance, when modifying a Mutex for writing a file, used as an argument in the NtCreateFile syscall and followed by the execution of the written file, we also modify that file everywhere else in the report.
Similarly, changes made Files and Keys categories are reflected accordingly; for further details on how the feature-space transformations we employ are functionality-preserving in the problem-space, please refer to \autoref{app:transformations}.

While we have shown how the modifications we consider are a) functionality preserving, and b) representative of the full extent of adversarial capabilities $\mathcal{C}$ in the problem-space, we still lack a concise methodology for applying them.
Next to \textit{what} manner to modify with, we additionally have the questions of \textit{where} and \textit{how}; we therefore require a comprehensive modification policy that can perform all of the above concurrently.
We grant these adversarial capabilities to our \ac{RL} agent by defining a multidiscrete\footnote{https://www.gymlibrary.dev/api/spaces/} policy.
This is a policy that at each time step selects three actions simultaneously: one action for each of the three aforementioned discrete categories.
Each action represents a different choice on the modifications aspect of the JSON report: what to do, where to do it, and how to do it.
Concretely, the three action categories are composed of the following discrete options to choose from:
\begin{enumerate}
    \item \textbf{What}: Add, Edit, X-Edit
    \item \textbf{Where}: One of the 13 categories described in \ref{sec:representation}.
    \item \textbf{How}: Random String, English Vocabulary, Target Vocabulary, Random Choice.
\end{enumerate}

In the first category, \textbf{Add} always selects one of the entries from the target class, based on their frequency of occurrence.
\textbf{Edit} replaces the entries in the report section that the second action points to (where), while keeping an incremental pointer for each section to avoid overwriting already edited entries. 
\textbf{X-Edit} replaces the entry or entries that the explainer (\autoref{sec:explainer}) returns as the most important.
The \textbf{Where} action points to the report section that the previously selected action will be performed.

\begin{table}[t]
\centering
\begin{tabular}{l|c|c|c}
    \toprule
    Category & Addition & Replacement & Removal\\
    \midrule
    - /Write/Read/Delete Files & \cmark & \cmark & \xmark \\
    - /Write/Read/Delete Keys & \cmark & \cmark & \xmark \\
    Executed Commands & \cmark & \cmark & \xmark \\
    Resolved APIs & \cmark & \xmark & \xmark \\
    Mutexes & \cmark & \cmark & \xmark \\
    Started/Created Services & \cmark & \cmark & \xmark \\
    \bottomrule
\end{tabular}
\caption{Feasible transformations by feature category}
\label{table:capabilities}
\end{table}

As for the four manners of replacement: \textbf{Random String} constructs a string of ASCII characters selected at random, \textbf{English Vocabulary} selects an English word at random, \textbf{Target Vocabulary} selects word from the vocabulary of the target class weighted by its frequency, while \textbf{Random Choice} selects one of the three previous choices at random for each step of building a longer string, e.g. a filepath.
The composition of these three actions results in a very capable policy that can modify reports of arbitrarily large size, with the added benefit that any modifications will always result to viable reports, as these would be generated from the problem-space and actual binaries.
Ideally, one could investigate an even more granular modification policy, one that is able to modify reports on the character rather than the word or token level as we currently do.
The reports in our dataset are on average 70K characters long, something that renders learning a character level policy on such a context length computationally intractable.

%% file: binaries.tex
\section{Adversarial Binaries}
\label{sebsec:binaries}

Our objective in this work is defensive: we are interested in generating adversarial malware in order to adversarially train the HMIL model.
Without loss of generality, we perform the functionality-preserving transformations in the feature space, while also introducing a proof of concept on how such adversarial variants of the malware can be constructed on the binary level.
We manually verified that modifications are consistent and functionality-preserving when performed on the feature space.

The modification policy of AutoRobust is constantly updated through interaction with the classifier.
This means that at each time step, we need to generate a new adversarial malware, then generate its dynamic analysis report, then classify it and repeat.
Besides the engineering effort required to do this in real-time, the computational overhead of this process would be inordinate.
The wall time formula in hours would thus be: $T = m * i * \frac{150}{3600}$, where $m$ is the modifications performed per iteration, $i$ is the number of iterations, and $150$ are the seconds that each binary is executed for.
In our evaluation $m = [50000, 200000]$ and $i = 15$, therefore $T = [31250, 125000]$.
Thus if performed on the problem space, a full adversarial training session would take on average 7 years to complete, on our equipment and research prototype.
In the feature space, as is currently performed, it takes on average 45 hours; a reduction by a factor of 1360.
In this context, the only feasible solution is to apply the transformations directly on the dynamic analysis reports, while preserving the problem space constraints and reducing the computational cost considerably.
AutoRobust then allows one to harden a malware classifier against evasion without incurring excessive computational costs.
We note here that whereas attackers \textit{must} generate realistic adversarial binaries, defenders do \textit{not} have to do that; rather, they merely have to generate feature spaces that correspond to realistic transformations.
This is an asymmetry that favors defenders as it makes adversarial training much more practical.

The problem-space modifications can occur either on the source code or the compiled binary.
For the former, as semantic equivalence is an undecidable problem, there are potentially unbounded ways to achieve the same functionality~\cite{baldan2021rice}; we therefore consider source code modifications out of scope.
Focusing on adversarial binaries instead, the capability we want to bestow to them is adaptive behavior, so that during their execution they can alter this behavior on demand, as it is observed by the dynamic analysis environment.
This can be achieved by the capacity to interrupt the functionality that their semantics dictate in order to interject an arbitrary number of randomly or deliberately selected API calls, or modify the arguments of the calls.
On Windows and the Portable Executable (PE) format this can be achieved via hooking the Import Address Table (IAT), which contains pointers to function addresses that can be called by the binary; the goal is to overwrite these addresses to point to adversary defined functions which execute the original plus additional functionality that they define.

Related work has demonstrated how to achieve adversarial behavior on the binary level~\cite{rosenberg2018generic, fadadu2019evading}.
The binary can packed with wrapper code that hooks all APIs and a configuration file that holds the modification logic---for AutoRobust that would be the adversarial policy $\pi_\theta^\mathcal{V}(a|s)$.
Alternatively, external calls to runtime libraries can be funneled through an indirect jump to the address specified in the IAT.
That jump points to the adversarial code consisted of the wrapped APIs, and each wrapped API calls the original API and before returning the return value, it applies the appropriate modifications.
In both cases to remain stealthy and preserve functionality, the additional calls should be provided with valid inputs and make sure there are no state changes on the memory and registers.

%% file: evaluation.tex
\section{Evaluation}
\label{sec:eval}

\begin{table*}[!ht]
    \centering
    \setlength{\tabcolsep}{0.3cm} % Adjust column spacing here
    \caption{Performance of AutoRobust (AR), gradient-based (GB) adversarial training, and the original vanilla HMIL model considering the two perturbation budgets, 1K and 2K. The metrics displayed are: Score (of the origin class), Attack Success Rate (ASR), Clean Accuracy (C.Acc), and Robust Accuracy (R.Acc), all reported on the same hold-out test set.}
    \begin{tabular}{c|c|r|r|r|r|r|r|c}
    \toprule
    \multirow{2}{*}{\textbf{Model}} & \textbf{Adv.} & \multicolumn{2}{c|}{\textbf{Score}} & \multicolumn{2}{c|}{\textbf{ASR (\%)}} & \multicolumn{2}{c|}{\textbf{R.Acc (\%)}} & \textbf{C.Acc (\%)} \\
    & \textbf{Training} & \multicolumn{1}{c}{1K} & \multicolumn{1}{c|}{2K} & \multicolumn{1}{c}{1K} & \multicolumn{1}{c|}{2K} & \multicolumn{1}{c}{1K} & \multicolumn{1}{c|}{2K} & \\
    \midrule
    \multirow{3}{*}{Binary} & None & 0.31 & 0.23 & 73.8 & 80.2 & 26.2 & 19.8 & 99.0 \\
    & GB & 0.32 & 0.30 & 72.2 & 71.8 & 27.8 & 28.2 & 99.1 \\
    & AR & \textbf{1.00} & \textbf{0.98} & \textbf{0.0} & \textbf{1.6} & \textbf{100} & \textbf{98.4} & 99.1 \\
    \midrule
    \multirow{3}{*}{Multiclass} & None & 0.11 & 0.25 & 100 & 99.8 & 0.0 & 0.2 & 98.8 \\
    & GB & 0.12 & 0.12 & 99.5 & 99.8 & 0.5 & 0.2 & 98.8 \\
    & AR & \textbf{1.00} & \textbf{0.99} & \textbf{0.0} & \textbf{0.5} & \textbf{100} & \textbf{99.5} & 98.8 \\
    \bottomrule
    \end{tabular}
    \label{tbl:results}
\end{table*}

In this section, we elaborate on the practical details of our approach that aims to harden the \ac{HMIL} model against evasion, which has been trained on the reports produced by dynamic analysis.
To acquire adversarial malware to train with, we operate directly on these reports after having rigorously defined the permitted transformations that enforce all necessary constraints in the problem-space.

As elaborated in Section \ref{sec:threat}, we do not generate new variants of the binaries themselves, however we do test and confirm that these transformations are possible to perform by modifying the CAPEv2 monitor library, capemon \cite{capemon}.
Specifically, we modified its source code adding the logic to successfully perform the proposed transformations via intercepting the syscalls of interest, e.g. for files \textit{NtCreateFile}, \textit{NtOpenFile}, \textit{NtReadFile}, etc.
In the case of a substitution operation, we replace the argument of the intercepted syscall, while for addition we invoke new system calls, depending on the intended feature to add.
In the first scenario, in order to keep track of all the changes and do not affect the original functionality (further elaborated through an example in \autoref{app:transformations}) we use a shared memory structure that stores the mappings between original and modified entries.
Every time that a system call is invoked, we check if one of the arguments matches with the entries of our mapping and, if this is the case we subsequently propagate the modification.

\subsection{Evaluation Setup}
We assume an initial \ac{HMIL} model non-adversarially trained on dynamic analysis reports.
The explicit goal of the adversary then is to find the optimal policy that, given the current state of the report, will do those modifications that lead to evasion.
Optimal here is meant in a twofold way: in terms of a successful evasion \textit{and} in terms of the minimal amount of transformations performed.
This procedure is performed in an episodic manner, by first selecting a report belonging to the origin class and subsequently modifying it until it evades detection or the maximum number of modifications (the perturbation budget) is reached.
When this happens the \ac{RL} agent proceeds to a new report, and continues for a predefined total number of training steps; after that, all the resulting adversarial reports are appended to the adversarial dataset.

Subsequently, one epoch of adversarial retraining is performed by sampling from \textit{both} the original and the adversarial dataset.
The above steps constitute a \textit{single} iteration in the AutoRobust methodology, with 15 of them performed in total.
The schematically depict the full pipeline of our experimental evaluation approach in Figure \ref{fig:autorobust}, while more information on the HMIL model, the RL agents, and all hyperparameters are included in Appendix \ref{app:hyper}.

\subsection{Metrics}

To properly measure and assess the performance of our approach, we monitor 4 primary metrics:
\begin{itemize}
    \item \textbf{Score}: the probability that the HMIL model assigns to the \textit{origin} class at the end of each episode, averaged over all episodes.
    \item \textbf{ASR}: the Attack Success Rate is the percentage of episodes in one iteration that the malware successfully evades; $ASR = \frac{s}{n} \times 100\%$, with $s =$ \textbf{Score} $\leq 0.5$ and $n$ the total number of episodes. ASR is a very relevant metric to assess how capable a policy is in turning malware reports adversarial.
    \item \textbf{C.Acc}: the model accuracy on clean samples from the hold-out set, measured in \%.
    \item \textbf{R.Acc}: the model accuracy on adversarial samples from the hold-out set, that is \textit{after} retraining the model.
\end{itemize}

Between adversarial training iterations, the modification policy can widely vary, and thus the distribution of modifications and adversarial reports will vary accordingly.
Therefore, \textbf{R.Acc} \textit{is not} a true indication of the level of robustness: while \textbf{R.Acc} is indeed evaluated on the hold-out test set of adversarial examples, these are still generated from the same policy and will therefore come from the \textit{same} distribution.

The only way to acquire a representative level of robustness is through training a new adversarial policy from scratch, to properly assess if the model has any ``blindspots'' left.
This process is repeated until convergence; at its end, a final more extensive round of training is performed with a hold-out set, and the resulting \textbf{ASR} is the one reported.
As a baseline to compare our approach to, we additionally perform and evaluate \textit{gradient-based} adversarial training, considered as the state-of-the-art defensive approach for hardening models against evasion, irrespective of domain~\cite{madry2017towards}.

\subsection{Results}
In our experiments we evaluate the robustness of two types of classification models: one binary to decide between malware and goodware, and one multiclass to decide between malware family.
Our adversaries operate under 2 perturbation budgets, namely 1K and 2K, while each entry modification, be it addition or replacement, counts as one perturbation.
Given that the average size of a report in our dataset is $\sim$1K entries, these can be considered very generous constraints.
For the binary and multiclass settings, Table \ref{tbl:results} reports the performance of the original model as well as the two different kinds of adversarial training we evaluate: Gradient-Based (GB) and our approach AutoRobust (AR).
From our empirical results we can make the following observations:
\begin{itemize}
    \item[$\bullet$] While the initial non-adversarially trained model achieves excellent clean accuracy, it is excessively vulnerable to evasive attacks, often with as few as 10 modifications. This is a good indication that while the model initially picks up many useful discriminative features, these are brittle and decidedly not robust~\cite{ilyas2019adversarial}.
    \item[$\bullet$] Notably, gradient-based adversarial training shows little potential in defending against problem-space adversaries. This can be intuitively explained as perturbations based on the gradient of the model will frequently map to either invalid or uninformative areas of the feature-space that map to problem-space changes the adversary will be either unable or uninterested to do. In such domains it is therefore strictly preferable to generate adversarial examples using capabilities you want to defend against.
    \item[$\bullet$] During training we can observe that the ASR consistently drops over few iterations. After all 15 iterations are finished, the ASR in the hold-out test set is always 0, demonstrating that for this model and under these adversarial capabilities, it is conceivable to fully defend against the latter. 
    \item[$\bullet$] Furthermore, we confirm that our approach is robust to hyperparameter selection, which is described in Appendix \ref{app:hyper}. Fig. \ref{fig:training} plots the mean and variance of ASR, C.Acc, and R.Acc for the 15 iterations and over multiple runs with different hyperparameters. While both C.Acc and R.Acc consistently remain near 100\%, ASR gradually depletes to near 0.
    \item[$\bullet$] Surprisingly, adversarial training with GB and AR slightly \textit{increases} the C.Acc of the binary model. This might appear counter-intuitive at first, as robustness and performance are typically in trade-off; note however that we do 15 more iterations of training compared to the vanilla model, where adversarial training can also function as a form of regularization.
    \item[$\bullet$] Finally, the model's transition from being outright vulnerable to fully robust indicates that it is imperative to perform adversarial training in any security critical domain where adversarial agency should be considered a given.
\end{itemize}

\begin{figure*}
\centering
\begin{subfigure}{1\columnwidth}
  \centering
  \includegraphics[width=\columnwidth]{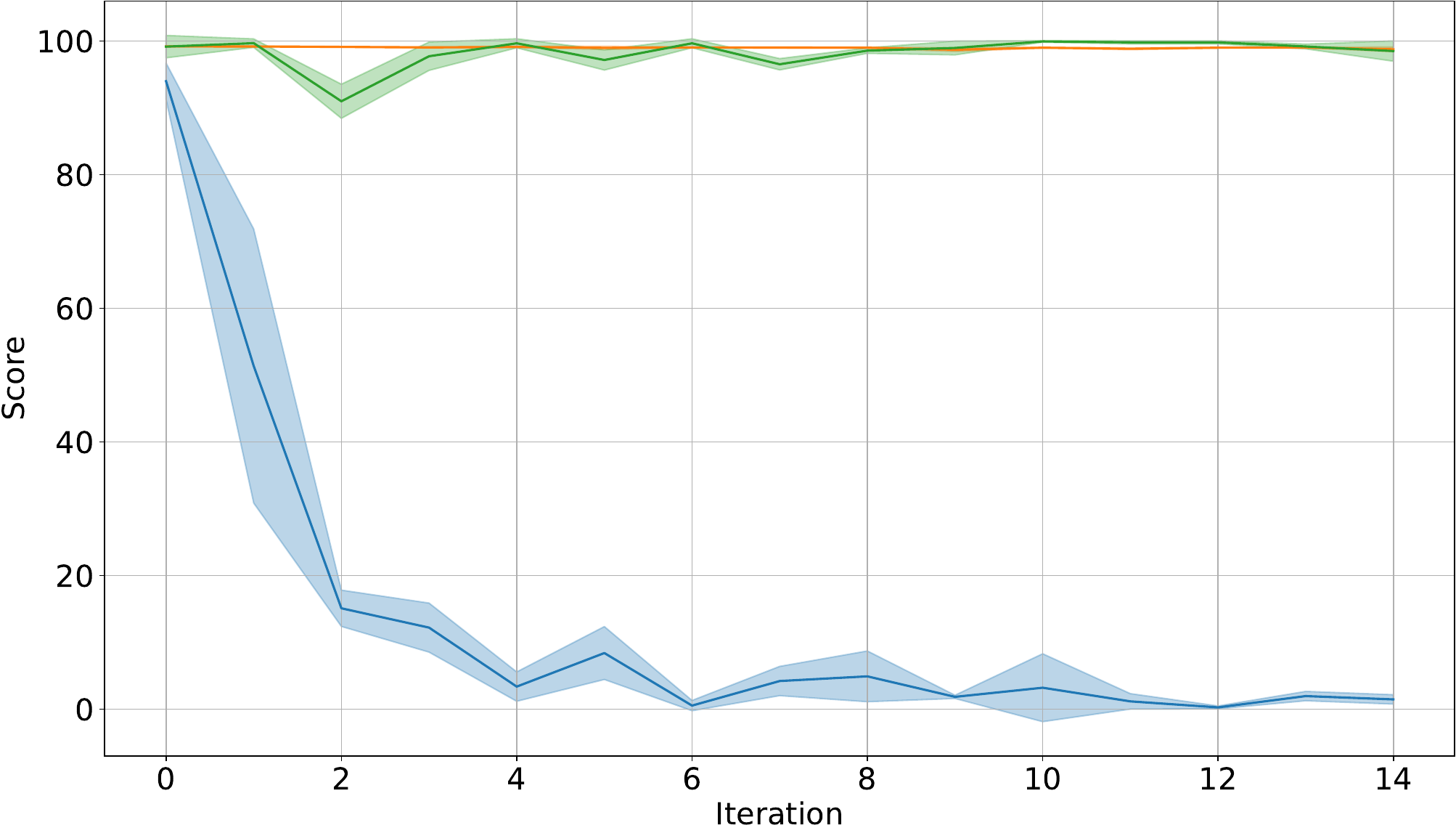}
  \caption{Binary classification}

\end{subfigure}
\begin{subfigure}{1\columnwidth}
  \centering
  \includegraphics[width=\columnwidth]{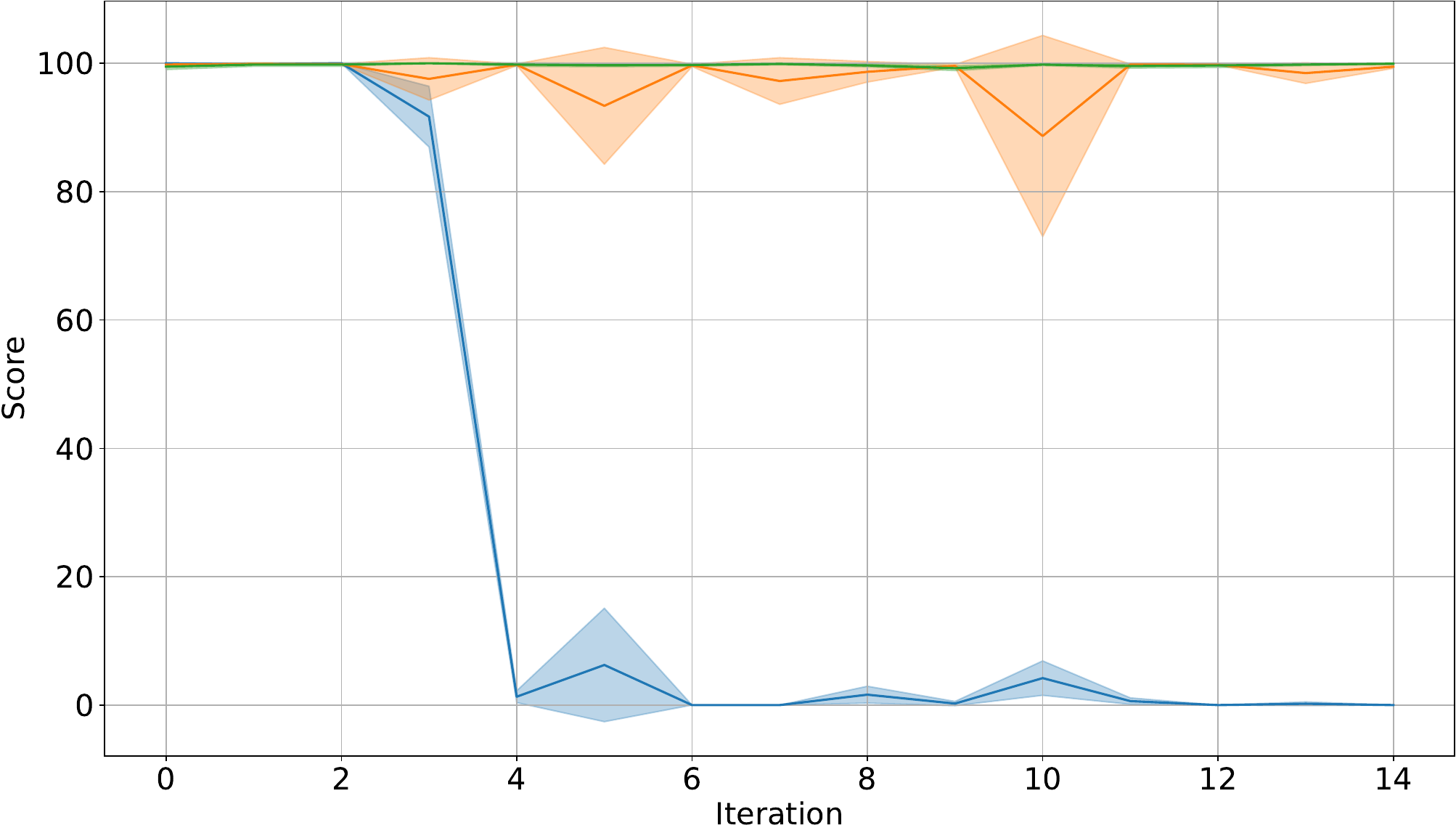}
  \caption{Multiclass classification}

\end{subfigure}
\caption{Progression of Attack Success Rate (blue), Clean Accuracy (yellow), and Robust Accuracy (green) over the 15 iterations of adversarial training. For each metric, the mean with one standard deviation are displayed, as we do multiple runs where hyperparameters and random seeds for selecting samples vary.}
\label{fig:training}
\end{figure*}

\begin{figure*}
\centering
\begin{subfigure}{1\columnwidth}
  \centering
  \includegraphics[width=\columnwidth]{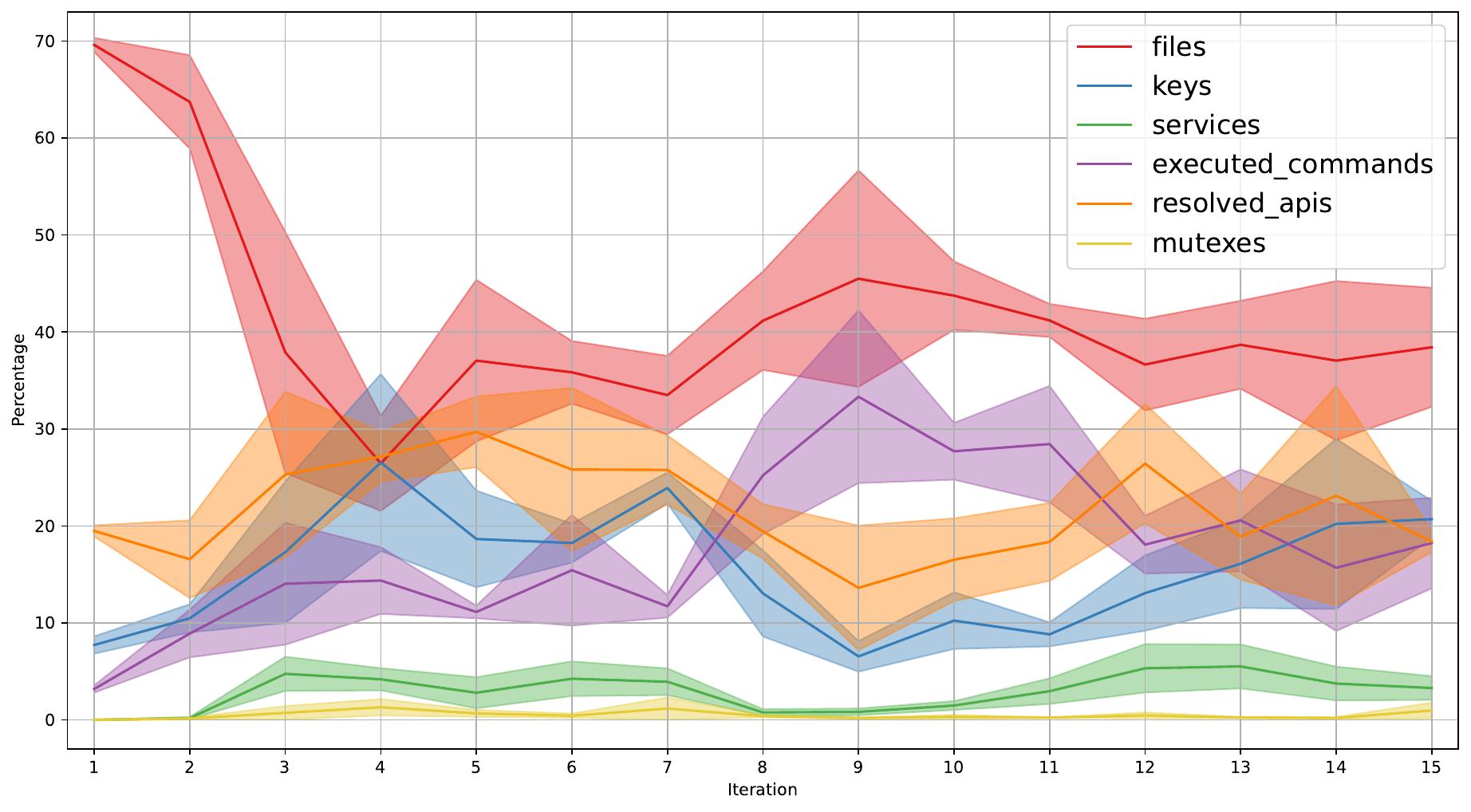}
  \caption{Explanations in binary evaluations.}

\end{subfigure}
\begin{subfigure}{1\columnwidth}
  \centering
  \includegraphics[width=\columnwidth]{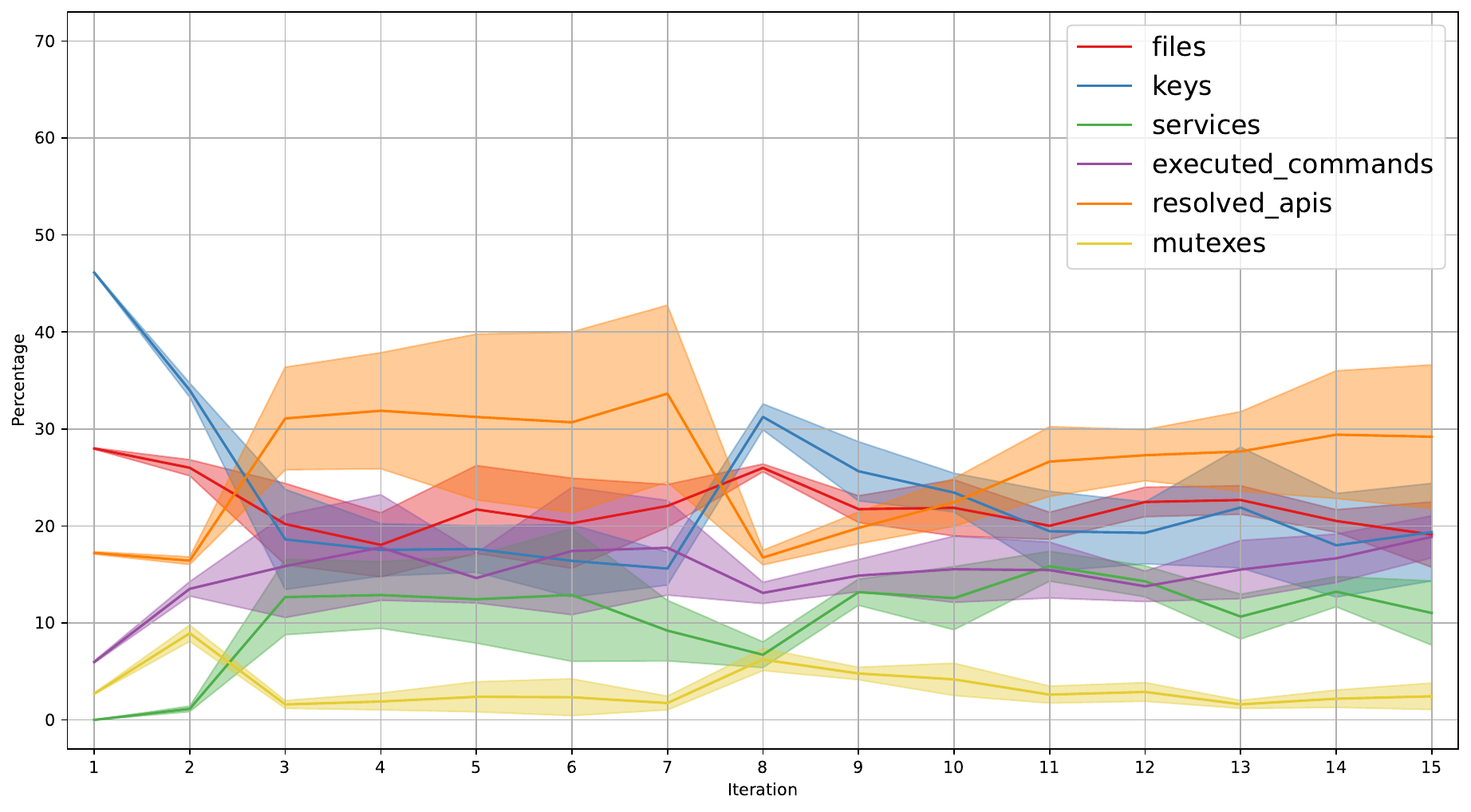}
  \caption{Explanations in multiclass evaluations.}

\end{subfigure}
\caption{The most frequent explanations as returned by the HMIL model explainer, over the 15 iterations. As for each iteration the number of episodes varies, we normalize explanations to a percentage over the whole iteration; mean values with one standard deviation are plotted.}
\label{fig:explanations}
\end{figure*}

\subsection{Explanations}
While deep learning obviates the laborious process of feature engineering and holds great promise in generalizing beyond its training data, we have to be judicious and inquisitive about what the model actually learns.
Spurious correlations are not only erroneous biases that a model can pick up when trained end-to-end on raw data, they additionally can be exploited by an attacker to evade the model.
It is thus crucial to remove such attack vectors at the time of training.
With AutoRobust we automate this demanding, human-in-the-loop process, by substituting the requirement for expert input on every potential shortcut feature, with asking a simple question: what are \textit{all} the feasible ways something could be different and still be the same thing?

To get a deeper understanding of how the feature importance for a model shifts as the latter is retrained, we collect the explanations within a single iteration -- that contains multiple episodes -- and evaluate how these explanations evolve over multiple iterations.
Figure \ref{fig:explanations} plots this progression.
We remind the reader that the explanations returned are not scored by importance, but instead consist of a discrete, variable size set with entries from the dynamic analysis report.

In the binary model, we observe that initially a disproportionate focus is put on files that later decreases, while registry keys and resolved APIs increase and remain important.
\textit{Mutexes} and \textit{services} represent a very small part of the explanations, unsurprisingly as they both have on average very few entries; what is surprising however is that \textit{executed commands} consistently rise in importance.
In the multiclass model we observe comparable patterns, with the notable exception that for adversarially retrained models \textit{resolved apis} are the most important explanations.
The latter is in line with our intuition that as the only category where entries cannot be replaced, it should be more informative when discriminating between malicious and benign binaries.

In conclusion, AutoRobust proves exceptionally effective in hardening a ML model against problem space attacks, and it can be viewed as an approach for first forcing and subsequently learning on the possible distribution shifts.
By changing in a report what could be different -- and \textit{only} what could \textit{realistically} be different -- without the label changing, it is a methodology for performing a counterfactual investigation on what is an essential feature and what is an artifact of the dataset.
In that way, AutoRobust shares similarities with (and inspiration from) dataset enhancement techniques, where transformations on the original dataset help the model generalise~\cite{cubuk2018autoaugment}.

%% file: discussion.tex
\section{Discussion}

When considering hardening ML-based models and adversarial training in general, gradient-based approaches are state-of-practice.
In our work we demonstrate that while this might be true for specific domains, it is not necessarily a general principle; in the dynamic analysis based malware detection we consider in this work this appears to not be the case.
Instead, our novel approach has proven exceedingly effective in hardening a model and defending against problem-space attacks -- attacks that are pervasive and realistic in real-world settings -- while gradient-based adversarial training showed very limited effectiveness.
This does not come as a complete surprise and has an intuitive explanation: in domains where the input space is unconstrained in size, but otherwise structured and constrained in adversarial capabilities, it is better to adversarially train by utilizing the latter two.
One can think of our approach as confining the adversarial search within the manifold that corresponds to realizable programs.

While in the binary setting the performance metrics follow a well-defined trend, the multiclass case shows a few aberrations.
As these co-occur with a drop in clean accuracy, we can partly attribute them to the ratio of clean and adversarial examples used for retraining and catastrophic forgetting.
This issue became much less prominent when this ratio was tweaked, while subsequent iterations showed conclusively that the model adapted successfully and has the capacity to learn both clean and adversarial versions of reports.
Overall, our evaluation demonstrates that AutoRobust is capable of zeroing out the success rate of attacks and is robust to hyperparameter selection; remarkably, it slightly improves the performance on clean data too.

\textbf{Limitations}. Our intention in this work was to focus on a real-world setting, hence we collaborated with an industrial partner that provided the ML model used as a component of their widely-known commercial antivirus.
This makes our results very relevant for real-world settings, but naturally our evaluation is carried out on the \ac{HMIL} model.
We have however supplemented this empirical part with an extensive theoretical analysis that demonstrates the generality of our approach for any model or context that exhibits a problem/feature-space gap; AutoRobust is agnostic to the specific problem-space capabilities that are used to generate adversarial variants.

Another potential limitation is that due to stringent testing protocols, our malware analysis has been conducted without any network connectivity.
Due to our compliance with safety guidelines, it is plausible that some of the analyzed binaries displayed a narrower extent of their full behavior due to lack of connectivity.
While our methodology does not involve direct modification of the binaries themselves, our transformations preserve functionality by design.
Nonetheless, we have verified their correctness on a handful of adversarial executions traces.
Specifically, in the proof of concept we developed (cf. Appendix \ref{app:transformations}) we manually verify these transformations are correctly performed on binaries under dynamic analysis and do not break the malware behavior or functionality.
Finally, due to the computational cost on our research infrastructure and the transformations being functionality-preserving, the AutoRobust pipeline is executed on the feature space.
This naturally limits our capacity to address or modify aspects of binary behavior that lie outside of this specific feature representation.

\textbf{Future Work}. Our AutoRobust framework, together with our p-Robustness formulation [\ref{th:pspr}] are pertinent to any domain that has a distinct gap between problem and feature spaces~\cite{dyrmishi2023empirical}.
It is thus a very promising approach to apply to other real-world settings and environments.
Additionally, learning modification policies at the character level proved a considerable technical challenge, especially when some reports span millions of characters.
Thus in our work we focused on more coarse controls, which when composed are nevertheless very powerful.
In the future, a more granular policy could control modifications at the character level.
Large Language Models (LLMs) are good candidates as they can implicitly learn the structure of a valid report, while their generation policy can be adapted to generate adversarial examples through \ac{RL} feedback~\cite{christiano2017deep}.

\section{Conclusion}

With ML-based detection of Windows malware on dynamic analysis environments becoming more prominent,  significant efforts have been dedicated to assessing the resilience and effectiveness of the ML classification against adversarial malware, that is variants that can evade detection.
In this work, we justify theoretically and empirically why adversarial attacks on malware should be performed through the problem space, and introduce \textit{AutoRobust}, a novel methodology based on \ac{RL} for hardening detection models against malware evasion.
As the information contained in dynamic analysis reports will most certainly include artifacts, \textit{AutoRobust} automates the identification and removal of spurious correlations and brittle features in the input space of the model which can be used for evasion: namely to generate adversarial malware.
By optimizing with \ac{RL} the process of discovering counterfactuals to what a naively trained model has learned, we can guarantee a probabilistic level of robustness against well-defined adversaries.
Notably, given a ML-based malware detection module of a commercial antivirus, we show how vulnerable it can be to adversarial malware, as well as how \textit{AutoRobust} can grant near perfect robustness to the model \textit{without} adversely affecting and even improving its performance on clean inputs.

%% file: appendix.tex
\appendix

\section{Problem-space transformations}
\label{app:transformations}
In Section \ref{sec:eval} we presented the full range of functionality-preserving transformations on the specific feature space we are investigating in this work.
As generating new binary variants for every modification performed is a considerable engineering effort and time consuming process, we developed a proof of concept (PoC) that covers the basic usage of every Windows API intercepted by capemon.

\begin{itemize}[leftmargin=*]
    \item In this PoC we implement the same exact logic that AutoRobust applies in the feature space.
    \item We create a shared memory table that keeps track of all the name changes across program execution, so whenever one of the APIs of interest is called its argument was replaced accordingly as indicated by the modification.
    \item This entry is then added to the shared table as a stage 1 modification; if there are successive modifications we do track them all for completeness.
    \item From that moment on, whenever another API call is invoked, we do check if the original name is in one of the arguments and replace it according to our substitution history.
    \item Depending on the API we employ different strategies for argument substitution. As we do not want to impact the original program behavior, we have to conform to the viable transformations for each API. For instance, creating a mutex allows for completely arbitrary names, for file generation we have to provide valid filepaths, and in command execution we need to ensure all previous modifications are reflected without introducing new random entries which could make the program crash. We preserve this reasoning in the feature space, as illustrated in \autoref{app:poc_report}.
\end{itemize}

Adapting our PoC to runtime modifications could include the aforementioned logic in a dll wrapped with the binary, as discussed in Section \ref{sebsec:binaries}.
Under the specifics of dynamic analysis and our use-case, given access to the decisions of the model $\mathcal{M}$ attackers can generate adversarial binaries in the problem space with considerable automation, but with less granularity than AutoRobust: whatever modifications they make, they have to be packaged in batch to a new binary that will be subsequently analyzed by the sandbox.
As the entries in the summary of the dynamic analysis report are not ordered, the placement of functions in adversarial binaries does not matter; as long as they return to the main functionality of the program with its internal state unaffected, these functions are allowed to interact with each other.

\subsection{Functionality Preservation in AutoRobust}

Here we outline in detail the modification procedure, which keeps track of the changed entries, and how it ensures that all changes propagate correctly throughout a dynamic analysis report. In order to do it, after the agent has performed a modification on the report, we seek for other occurrences of the same entry in different categories. In case we do, we perform the same modification on it and we do this for all the dynamic report categories \autoref{sec:representation}. 
In that way the modifications performed will reflect both a viable report variant in the feature-space, and a binary with its functionality preserved in the problem-space.
For this detailed explanation, we consider a malicious program that wants to modify the firewall rules.
To do that, it needs to invoke \textbf{netsh} Windows utility, something we expect to appear in the category \textit{executed commands}: 

\begin{lstlisting}[language=bash]
netsh firewall add allowedprogram "C:\Users\John\AppData\Roaming\malicious.exe" "malicious.exe" ENABLE
\end{lstlisting}

Changing the command itself or its could inadvertently affect the original malware behavior. However, renaming the file from "malicious.exe" to another name, provided the new name is consistently referenced, would not affect the malware's functionality.
An initial modification could be:

\begin{lstlisting}[language=bash]
netsh firewall add allowedprogram "C:\Users\John\AppData\Roaming\hello.exe" "hello.exe" ENABLE
\end{lstlisting}

Following a modification on the executed commands, our transformation logic performs a sweep of the entire report to identify all places that it is referenced; subsequently, every occurrence of "malicious.exe" is substituted with "hello.exe", for example corresponding entries within both the \textit{Files} and \textit{Write Files} sections.
Such entries indicate not only the generation of the file but also the activities related to content modification associated with it and in order to be consistent we change those entries as well.
Additionally, our investigation unveils instances within the \textit{executed commands} section that indicate subsequent efforts to execute the binary after adjustments had been made to the firewall settings.
Our transformations ensure the modification of the filename in this case as well, in that way preserving the original functionality.

This example helps to demonstrate how modifications are propagated around the executed commands category. The same reasoning is applied to the performed modifications that need to be reflected in other categories; for example changes in \textit{Write Keys} are reflected the \textit{Keys} section.
Because we do not modify the binaries themselves, we can guarantee that for the granularity provided by our current feature space the transformations are coherent, working at the same abstraction level.

\subsection{Dynamic Analysis Reports}
\label{app:poc_report}
To illustrate the feature space we are working with, and also demonstrate how effortless evading the HMIL model initially is, we include an example dynamic analysis report as well as its adversarial variant.
\autoref{lst:example_report_or} shows the original report as generated through execution in our sandbox.

\begin{lstlisting}[language=json, label=lst:example_report_or, caption=Original report entries]
{
  "summary": {
    "files": [
      "C:\\Users\\John\\AppData\\Local\\Temp\\\\xb0\\xa1\\xb0\\xa1\\xb0\\xa1\\xb0\\xa1\\xb0\\xa1\\xb0\\xa1\\xb0\\xa1\\xb0\\xa1\\xb0\\xa1\\xb0\\xa1\\xb0\\xa1\\xb0\\xa1\\xb0\\xa1\\xb0\\xa1\\xb0\\xa1\\xb0\\xa1\\xb0\\xa1\\xb0\\xa1\\xb0\\xa1\\xb0\\xa1\\xb0\\xa1\\xb0\\xa1\\xb0\\xa1\\xb0\\xa1\\xb0\\xa1\\xb0\\xa1\\xb0\\xa1\\xb0\\xa1\\xb0\\xa1\\xb0\\xa1\\xb0\\xa1\\xb0\\xa1\\xb0\\xa1",
      "C:\\Users\\John\\AppData\\Local\\Temp\\03cb16b5aa21f4b0a10a8e3e.exe",
      "C:\\vbxjf.exe"
    ],
    "read_files": [
      "C:\\Users\\John\\AppData\\Local\\Temp\\03cb16b5aa21f4b0a10a8e3e.exe"
    ],
    "write_files": [
      "C:\\vbxjf.exe"
    ],
    "delete_files": [],
    "keys": [
      "DisableUserModeCallbackFilter",
      "HKEY_LOCAL_MACHINE\\Software\\Microsoft\\Windows NT\\CurrentVersion\\GRE_Initialize",
      "HKEY_LOCAL_MACHINE\\SOFTWARE\\MICROSOFT\\WINDOWS NT\\CURRENTVERSION\\GRE_Initialize\\DisableMetaFiles"
    ],
    "read_keys": [
      "DisableUserModeCallbackFilter",
      "HKEY_LOCAL_MACHINE\\SOFTWARE\\MICROSOFT\\WINDOWS NT\\CURRENTVERSION\\GRE_Initialize\\DisableMetaFiles"
    ],
    "write_keys": [],
    "delete_keys": [],
    "executed_commands": [
      "c:\\vbxjf.exe"
    ],
    "resolved_apis": [
       "kernel32.dll.VirtualAlloc",
      "kernel32.dll.VirtualProtect",
      "kernel32.dll.VirtualFree",
      "kernel32.dll.LoadLibraryA",
      "kernel32.dll.GetProcAddress",
      "kernel32.dll.ExitProcess",
      "msvcrt.dll.atoi",
      "shlwapi.dll.PathFileExistsA",
      "user32.dll.wsprintfA",
      "kernel32.dll.OpenProcess",
      "kernel32.dll.VirtualAllocEx",
      "kernel32.dll.WriteProcessMemory",
      "kernel32.dll.WaitForSingleObject",
      "kernel32.dll.VirtualFreeEx",
      "kernel32.dll.GetProcessHeap",
      "kernel32.dll.GetModuleHandleA",
      "kernel32.dll.HeapAlloc",
      "kernel32.dll.HeapFree",
      "kernel32.dll.IsBadReadPtr",
      "kernel32.dll.DeleteFileA",
      "kernel32.dll.GetModuleFileNameA",
      "kernel32.dll.CloseHandle",
      "kernel32.dll.ReadFile",
      "kernel32.dll.GetFileSize",
      "kernel32.dll.CreateFileA",
      "kernel32.dll.WriteFile",
      "kernel32.dll.CreateProcessA",
      "kernel32.dll.GetStartupInfoA",
      "kernel32.dll.Sleep",
      "kernel32.dll.FreeLibrary",
      "msvcrt.dll.strchr",
      "msvcrt.dll._CIfmod",
      "user32.dll.TranslateMessage",
      "user32.dll.DispatchMessageA",
      "user32.dll.GetMessageA",
      "user32.dll.MessageBoxA",
      "user32.dll.PeekMessageA",
      "user32.dll.GetWindowThreadProcessId"
    ],
    "mutexes": [],
    "created_services": [],
    "started_services": []
  }
}
\end{lstlisting}

On the other hand, \autoref{lst:example_report_adv} shows the same report after enough transformations have been applied so that it has successfully evaded the target classifier, i.e. classified as goodware.
For simplicity, we report only the entries that have been modified.
We can observe that our transformations change consistently the name of the files and the directory paths in which these are saved, propagating also the modification to \textit{executed commands} section.
Then, our \ac{RL} agent also adds a mutex, which has no impact on the actual behavior of the original binary.

\begin{lstlisting}[language=json,label=lst:example_report_adv, caption=Modified report entries]
{
  "summary": {
    "files": [
      "C:\\Users\\John\\AppData\\Local\\Saxon\\Temp\\sorage\\shrinal\\(*@\linebreak@*)podostemad\\hired\\oxyphenyl\\demonocracy\\Lapsana\\Kamasin\\(*@\linebreak@*)menoschesis\\spinnerular\\symptomless\\unknelled\\Cladocera\\lame\\(*@\linebreak@*)suspensively\\uncomparably\\Utopianize\\demigroat\\ovistic\\equalize\\(*@\linebreak@*)lounger\\staghunt\\Cascadia\\head\\overrule\\chichimecan\\(*@\linebreak@*)Edestosaurus\\berrigan\\booger\\cinurous\\unhurt\\gollar\\uraniid\\(*@\linebreak@*)blousing\\lunge\\recrease\\rage\\limy\\predefense\\spadger\\(*@\linebreak@*)Wykehamist\\taxidermist\\outcrier\\patchwork\\specular\\huffishness\\(*@\linebreak@*)tarnal\\cabin\\tenementary\\rectalgia\\cataplasis\\flavor\\(*@\linebreak@*)sprank\\rigor\\unwrinkle\\partitionary\\dancery\\demihigh\\(*@\linebreak@*)warrantee\\sherifate\\thereout\\fourteenthly\\mousefish\\(*@\linebreak@*)unfairylike\\tenaille\\iodobromite\\octoglot\\Cartist\\(*@\linebreak@*)heartscald\\wellsite\\triphony\\picturely\\zoning\\deal\\pyelectasis\\(*@\linebreak@*)unselfish\\marshite\\truckful\\beworn\\thriller\\divaricately\\Earnie",
      "C:\\Users\\horseback\\zeuglodont\\supping\\aphasic\\institor.exe",
      "C:\\surgeproof\\marionette.exe"
      
    ],
    "read_files": [
      "C:\\Users\\horseback\\zeuglodont\\supping\\aphasic\\institor.exe"
    ],
    "write_files": [
      "C:\\surgeproof\\marionette.exe"
    ],
    "delete_files": [],
    "keys": [
      ...
    ],
    "read_keys": [
      ...
    ],
    "write_keys": [],
    "delete_keys": [],
    "executed_commands": [
      C:\\surgeproof\\marionette.exe}
    ],
    "resolved_apis": [
      ...
    ],
    ("mutexes": ["Asjk"],
    "created_services": [],
    "started_services": []
  }
}
\end{lstlisting}

\section{HMIL \& RL Hyperparameters}
\label{app:hyper}

The model is retrained for a total of 15 iterations on the new adversarial examples generated, where in each iteration agents are trained from scratch.
We cold-start new agents because agents from previous iterations often failed to evade the retrained model; this strongly indicates that knowing the capability set $\mathcal{C}$ is not sufficient to defend, as multiple different policies -- and thus multiple different distributions -- can result from it.
We emphasize that robust accuracy is not a true indication of actual robustness; while evaluated on a test set, this comes from the \textit{same} distribution due to being generated by the same agent.
The realistic level of robustness is evaluated only through a final round of attacks, with starting reports neither the model nor the agent have seen before.
As is best practice in \ac{AML}, candidate reports are only those that are correctly classified by the model.

\textbf{States \& Rewards}. For all agents, the state representation $s$ for a sample $x$ is its 32-dimensional embedding on the last layer before the fully connected of \ac{HMIL} model. 
The reward functions we evaluate with are straightforward: we reward the score decrease in the source class and its increase in the target class, while optionally penalizing the total number of steps and any disallowed moves -- which if selected, are never allowed to go through.
Table \ref{tbl:hyper} reports on all hyperparameters -- or ranges -- that we experimented with.
Throughout our experiments, we observe that results are robust to hyperparameter selection.

\begin{table}[h!]
\centering
\renewcommand*{\arraystretch}{1.05}
\caption{Hyperparameters.}
\begin{tabular}{|r|r|r|}
\toprule
\textbf{Hyperparameter} &\bf Binary &\bf Multiclass\\
\midrule
learning rate & e-3 -- 3e-3 & e-3 -- 3e-3 \\
steps per episode & 1000 -- 2000 & 1000 -- 2000 \\
steps per iteration & 5e4 -- 3e5 & 2e4 -- 2e5\\
iterations & 15 & 15\\
batch size RL& 32 & 32\\
batch size HMIL& 128 & 128\\
\bottomrule
\end{tabular}
\label{tbl:hyper}
\end{table}